\documentclass[pdflatex,sn-mathphys-num]{sn-jnl}

\usepackage{graphicx}
\usepackage{amsmath,amssymb,amsfonts}
\usepackage{amsthm}
\usepackage{booktabs}
\usepackage{enumitem}
\usepackage{subcaption}
\usepackage{algorithm}
\usepackage{algpseudocode}

\theoremstyle{thmstyleone}
\newtheorem{theorem}{Theorem}
\newtheorem{lemma}[theorem]{Lemma}

\theoremstyle{thmstyletwo}

\hypersetup{
  pdftitle={A Fast Greedy Algorithm with a 0.45-Approximation Guarantee for the Airplane Refueling Problem},
  pdfauthor={Oleg Lomachenko and Arnold Enikeev},
  pdfsubject={Approximation algorithms for the Airplane Refueling Problem},
  pdfkeywords={airplane refueling problem, scheduling, greedy algorithms, approximation algorithms, dynamic convex hull}
}

\begin{document}

\title[Fast greedy algorithm for ARP]{A Fast Greedy Algorithm
with a $0.45$-Approximation Guarantee
for the Airplane Refueling Problem}

\author*[1]{\fnm{Oleg} \sur{Lomachenko}}
\email{lomachenko.science@gmail.com}
\author[1]{\fnm{Arnold} \sur{Enikeev}}

\affil*[1]{\orgname{Goodsforecast},
  \orgaddress{\city{Moscow}, \country{Russia}}}

\abstract{
We study a natural greedy algorithm for the Airplane Refueling Problem (ARP),
based on a well-known priority rule. We show that this algorithm achieves
an approximation guarantee of $9/20$, noticeably improving the previously
known $1/12$ guarantee for a fast algorithm. We also show that the same
priority rule can be implemented in $O(n\log n)$ time and $O(n)$ space using a
dynamic convex-hull data structure. These results provide a worst-case
approximation guarantee for an existing greedy priority rule and establish
an efficient implementation. In the real-RAM model, the
running-time bound is independent of the magnitude of the processing times.
The resulting method is also suitable
as a fast primal routine inside more complex exact and approximation algorithms
for large instances. These results narrow the gap between theoretically strong
but complex approximation schemes and lightweight algorithms that can be used
directly in computational pipelines.
The algorithm combines a simple selection rule with a provable performance
bound, making it a useful baseline for further algorithmic studies of ARP.
}

\keywords{airplane refueling problem, scheduling,
greedy algorithms, approximation algorithms, dynamic convex hull}

\maketitle

\section{Introduction}
\label{sec:introduction}

The Airplane Refueling Problem (ARP) is a classical combinatorial
optimization problem dating back to the puzzle of Gamow and
Stern~\cite{GamowStern1958}. Early mathematical studies of the range of a
fleet of aircraft include Franklin~\cite{Franklin1960} and
Goldman~\cite{Goldman1961}, with the latter deriving optimal policies for
fleets with equal fuel consumption rates or equal fuel capacities.
In its modern formulation, it can naturally be
viewed as a single-machine sequencing problem.
Let $$J=\{1,\ldots,n\}$$
be a set of jobs. Each job $j\in J$ has a processing time
$$    p_j>0$$
and a weight $$w_j>0$$
For a permutation of jobs
$\pi=(\pi_1,\ldots,\pi_n),$
let
$$C_{\pi_k}=\sum_{\ell=1}^{k} p_{\pi_\ell}$$
denote the completion time of job $\pi_k$. \newline
The goal is to find a permutation
$\pi$ maximizing
$$F(\pi)=\sum_{j=1}^{n}\frac{w_j}{C_j}.$$
Hence, the problem is equivalent to finding
$$\max_{\pi\in S_n} \sum_{j=1}^{n}\frac{w_j}{C_j}$$
$S_n$ here is the set of all permutations of the $n$ jobs. In the original interpretation, $p_j$ corresponds to the fuel consumption rate of
aircraft $j$, and $w_j$ represents its fuel capacity~\cite{GamzuSegev2019,Vasquez2015,LiEtAl2019}.
\newline
Despite its simple formulation, the ARP remains a nontrivial optimization
problem. Its exact computational complexity has been explicitly discussed in
the literature~\cite{Woeginger2010,GamzuSegev2019,GmeinerSchulz2026}, while
closely related and more general models, including robust appointment
scheduling with heterogeneous costs and block stacking with counterbalancing,
are known to be computationally hard~\cite{SchulzUdwani2019,GmeinerSchulz2026}.
In this paper we focus on approximation guarantees and fast implementation of
a natural greedy rule for the ARP.
\newline
Regarding approximation, the problem has been investigated more extensively. Gamzu and Segev~\cite{GamzuSegev2019} developed a PTAS for the
ARP. However, the authors of the paper point out that this construction involves heavy guessing and rounding steps, potentially resulting in an impractical runtime. As a simple alternative, they
proposed a fast algorithm guaranteeing
$$\frac{1}{12}\,\operatorname{OPT}$$
The computational complexity of this algorithm is
$$O(n\log C),\qquad C=\sum_{j=1}^{n}p_j$$
Therefore, for the ARP, there is a notable gap between the strong theoretical guarantees of the PTAS and the guarantees provided by practically implementable algorithms.
\newline
This naturally leads us to consider a particularly simple greedy algorithm. Suppose
that, before the next selection, the total processing time of all previous jobs is $S$. For every job that has not yet been selected,
define
$$\phi_j(S) = \frac{w_j}{p_j(S+p_j)}.$$
At each step, the algorithm selects a job with the maximum value of
$\phi_j(S)$.
\newline
This rule is not new in the ARP literature. Li et
al.~\cite{LiEtAl2019} use the same function when constructing potential
schedules consistent with known dominance conditions. Moreover, this rule
arises naturally from the analysis of a local interchange of two adjacent
jobs. To the best of our knowledge, no worst-case approximation guarantee
has previously been established for this specific adaptive greedy rule for ARP.
The ARP is also explicitly mentioned by Bansal et al.~\cite{BansalEtAl2017}
in their study of local and global precedence relations for nonlinear
single-machine scheduling, where local precedence is defined through
adjacent interchanges, the same type of exchange underlying the greedy rule
studied here.
Lin and He~\cite{LinHe2026} also study precedence criteria and propose
a gradient-based scheduling heuristic for the ARP using neighborhood
exchange arguments.
\newline
Further motivation comes from the closely related
\emph{Robust Appointment Scheduling} (RAS) problem. For the variant in which
the job order is chosen simultaneously and the underutilization costs are
homogeneous, Schulz and Udwani~\cite{SchulzUdwani2019} studied a closely
related locally optimal algorithm based on the rule
$$\frac{\Delta_j}{o_j\bigl(o_j+o(S)+u\bigr)}$$
After normalization, one may set $u=1$. Therefore, this criterion differs from the
greedy criterion for the ARP only by an additional positive shift in the
denominator. For the corresponding algorithm, the authors obtained a strong
constant-factor guarantee for the RAS minimization objective, with an
approximation ratio of approximately
$$1.0604$$
Recently, Gmeiner and Schulz~\cite{GmeinerSchulz2026} formalized the
connection between the ARP and the corresponding version of RAS and showed
that the two problems are equivalent with respect to the existence of a
polynomial-time exact algorithm. This equivalence, however, does not preserve
multiplicative approximation guarantees. The reason is the
additional shift in the denominator of the RAS objective. The authors
explicitly note that approximation guarantees cannot be transferred directly between these two problems.
\newline
Consequently, even a very strong guarantee for the locally optimal algorithm
in RAS does not by itself answer whether the analogous greedy algorithm has
any constant-factor guarantee for the ARP. This is the main motivation for
our work.
\newline
We show that such a guarantee indeed exists. For the greedy algorithm
described above, we prove a constant-factor approximation guarantee of at least
$45\%$ of the optimal value:
$$G \geq \frac9{20}\,\operatorname{OPT}$$
This simple locally optimal algorithm is guaranteed to retain
at least $45\%$ of the optimal objective value, substantially improving upon
the known $1/12$ guarantee for the ARP.
\newline
Moreover, we show that quadratic running time is not inherent in this greedy rule. 
\newline Indeed,
$$\frac{1}{\phi_j(S)} = \frac{p_j}{w_j}S + \frac{p_j^2}{w_j}$$
so each job defines a linear function of $S$, and selecting the next job can
be reduced to an extreme-point query over the corresponding lines. Using the
semi-dynamic convex-hull of Hershberger and
Suri~\cite{HershbergerSuri1992}, we implement the entire algorithm in
$$O(n\log n)$$
time instead of the naive $O(n^2)$.
\newline
This paper provides two related results: we establish a constant-factor
approximation guarantee for a natural greedy algorithm and accelerate its
implementation to $O(n\log n)$. Taken together, these results establish that a
simple locally optimal approach to the ARP has considerably stronger
theoretical guarantees than were previously known.

\section{Adjacent-Interchange Criterion}
Adjacent-interchange arguments are classical in single-stage scheduling.
Smith~\cite{Smith1956} formulated a general sufficient condition under
which an ordering can be justified by comparing adjacent pairs and showed
that, in important cases, such comparisons reduce to ordering by an
individual index.

Consider two adjacent jobs $a$ and $b$. Suppose that the jobs preceding them have total processing time $S$. Compare the two local orders $a,b$ and $b,a$.
Their contributions to the objective are
$$\frac{w_a}{S+p_a} + \frac{w_b}{S+p_a+p_b}
\qquad \text{and} \qquad \frac{w_b}{S+p_b}+ \frac{w_a}{S+p_a+p_b} \text{  respectively.}$$
The completion times of all subsequent jobs are independent of the order. Hence
\begin{equation}
\begin{aligned}
&F(\ldots,a,b,\ldots)-F(\ldots,b,a,\ldots)\\
&\quad=\frac{w_a}{S+p_a} + \frac{w_b}{S+p_a+p_b} - \frac{w_b}{S+p_b}-\frac{w_a}{S+p_a+p_b}\\
&\quad=\frac{p_ap_b}{S+p_a+p_b}
\left(\frac{w_a}{p_a(S+p_a)}-\frac{w_b}{p_b(S+p_b)} \right)
\end{aligned}
\label{eq:swap}
\end{equation}
\newline
Define
\begin{equation}
    \phi_i(S)=\frac{w_i}{p_i(S+p_i)}.
    \label{eq:phi}
\end{equation}
\newline
Equation \eqref{eq:swap} implies the following criterion:

\begin{lemma}[Adjacent-interchange criterion]\label{lem:adjacent-interchange}
The order $a,b$ is no worse than the order $b,a$ if and only if
$$\phi_a(S)\ge \phi_b(S)$$
\end{lemma}
\begin{proof}
By \eqref{eq:swap}, the difference between the two objective values is
$p_ap_b/(S+p_a+p_b)$ times $\phi_a(S)-\phi_b(S)$.
Since $p_a,p_b>0$ and $S\ge0$, this factor is strictly positive,
which proves the equivalence.
\end{proof}
\section{The Greedy Algorithm}
Lemma~\ref{lem:adjacent-interchange} leads to a greedy algorithm.
This selection rule is used by Li et al.
\cite{LiEtAl2019} in their construction of a partial potential schedule.
\newline
\textbf{Greedy algorithm.}
Initialize the accumulated processing time $S$:
$$S:=0$$
Initialize set of remaining jobs:
$$R:=\{1,\ldots,n\}$$
and an empty permutation $\pi^G$. \newline While $R\neq\varnothing$ find $\psi$:
\begin{equation}
    \psi\in\mathop{\mathrm{arg\,max}}_{i\in R}\phi_i(S)=
    \mathop{\mathrm{arg\,max}}_{i\in R}
    \frac{w_i}{p_i(S+p_i)}
    \label{eq:greedy}
\end{equation}
append $\psi$ to the end of $\pi^G$, then update:
$$S\leftarrow S+p_\psi \qquad \text{and} \qquad R\leftarrow R\setminus\{\psi\}$$
Denote the resulting solution $G=F(\pi^G)$.
\begin{lemma}[Local optimality]
The schedule $\pi^G$ is locally optimal: no adjacent interchange increases the objective value $F$.
\end{lemma}

\begin{proof}
Consider two adjacent jobs $a,b$ in the greedy schedule. Let $S$ be the total processing time of the jobs preceding $a$. When $a$ was selected, $b$ was still available. By the greedy rule \eqref{eq:greedy}:
$$\phi_a(S)\ge \phi_b(S)$$
By Lemma~\ref{lem:adjacent-interchange}, the order $a,b$ is no worse than the order $b,a$. Since $a$ and $b$ were chosen arbitrarily, this completes the proof.
\end{proof}
\section{Approximation Guarantee}
The guarantee below concerns the schedule constructed by the greedy rule,
not an arbitrary schedule that is locally optimal under adjacent interchanges.
\newline
We begin with an integral representation of the objective through the
total weight of completed jobs. This is analogous to the area-under-the-curve
representation through unfinished weight used by Megow and
Verschae~\cite{MegowVerschae2018} in scheduling with varying machine speed,
with the kernel $t^{-2}$ accounting for our reciprocal completion-time objective.

Let $\pi^G$ and $\pi^*$ denote the greedy and optimal schedules, respectively,
and let $G=F(\pi^G)$, $\operatorname{OPT}=F(\pi^*)$, and $P=\sum_j p_j$.
\newline
For any schedule $\pi$, define
$$W_\pi(t)=\sum_{j:C_j^\pi\le t}w_j$$
Since $\int_{C_j^\pi}^P w_jt^{-2}\,dt=w_j/C_j^\pi-w_j/P$,
and the last term is independent of the schedule,
\begin{equation}
\operatorname{OPT}-G=\int_0^P
\frac{W_{\pi^*}(t)-W_{\pi^G}(t)}{t^2}\,dt
\label{eq:gap}
\end{equation}
\newline
For each job $a$ in the greedy schedule, let $S_a$ be its starting time
and define
$$q(t)=\phi_a(S_a)=\frac{w_a}{p_a(S_a+p_a)},
\qquad t\in[S_a,S_a+p_a)$$
Since the greedy priority of every remaining job decreases with time,
$q$ is nonincreasing. Indeed, after a greedy step the selected job is removed,
the value of $S$ increases, and for every remaining job
$w_j/(p_j(S+p_j))$ can only decrease. Moreover,
\begin{equation}
\int_0^P q(t)\,dt = \sum_a\frac{w_a}{C_a^G}=G
\label{eq:q-integral}
\end{equation}
\newline
Let $M(t)=\{j:C_j^*\le t<C_j^G\}$.
If the greedy schedule processes job $a$ at time $t$, every $j\in M(t)$
was available when $a$ was selected. Hence
$$q(t)=\phi_a(S_a)\ge\phi_j(S_a)
\ge\frac{w_j}{p_j(t+p_j)}$$
so $w_j\le q(t)p_j(t+p_j)$. Hence,

$$W_{\pi^*}(t)-W_{\pi^G}(t)
\le\sum_{j\in M(t)}w_j
\le q(t)\sum_{j\in M(t)}p_j(t+p_j)$$
Defining
$$h(t)=\sum_{j\in M(t)}
\left(\frac{p_j}{t}+\frac{p_j^2}{t^2}\right),\qquad t>0,$$
and setting $h(0)=0$,
we obtain
\begin{equation}
\operatorname{OPT}-G\le\int_0^P q(t)h(t)\,dt
\label{eq:gap-bound}
\end{equation}
\begin{lemma}
For every $T>0$,
$$\int_0^T h(t)\,dt\le\frac{11}{9}T$$
\end{lemma}
\begin{proof}
Within this proof, relabel the jobs in their optimal processing order,
so that $C_j^*=\sum_{i=1}^j p_i$, and set $C_0^*=0$.
If $T<C_1^*$, then $M(t)$ is empty for $0\le t\le T$ and the claim is immediate.
Hence assume $T\ge C_1^*$ and let $m=\max\{j:C_j^*\le T\}$.
\newline
Dropping the condition $t<C_j^G$ in the definition of $M(t)$ gives
\begin{equation}
\int_0^T h(t)\,dt
\le\sum_{j:C_j^*\le T}
\left[
p_j\ln\frac{T}{C_j^*}
+p_j^2\left(\frac1{C_j^*}-\frac1T\right)
\right]
\label{eq:h-prefix}
\end{equation}
\newline
Consider a job $j\le m$ in the optimal schedule and set
$a=C_{j-1}^*$, $b=C_j^*$, $p_j=b-a$, and $r=a/b$. Since $p_j>0$, we have
$a<b$, hence $0\le r<1$.
With expressions at $r=0$ interpreted by continuity, since $p_j=b(1-r)$ and $a=br$,
$$\int_a^b\ln\frac{T}{t}\,dt-p_j\ln\frac{T}{b}
=\int_a^b\ln\frac{b}{t}\,dt
=b(1-r+r\ln r)$$
Since $b=C_j^*\le T$,
$$p_j^2\left(\frac1b-\frac1T\right)
\le \frac{p_j^2}{b}=b(1-r)^2$$
Hence
\begin{align*}
p_j\ln\frac Tb+p_j^2\left(\frac1b-\frac1T\right)
-\int_a^b\ln\frac Tt\,dt 
\le br(r-1-\ln r)=p_j\psi(r)
\end{align*}
where
$$\psi(r)=\frac{r(r-1-\ln r)}{1-r},
\qquad 0<r<1,$$
with $\psi(0)=0$ by continuity.
\newline
We claim that $\psi(r)\le 2/9$ for every $0<r<1$. Set
$x=1/r>1$. Then
$$\psi(r)=\frac{r(r-1-\ln r)}{1-r} = \frac{\frac1x\left(\frac1x-1+\ln x\right)}{1-\frac1x} = \frac{\ln x+\frac1x-1}{x-1}$$
Hence $\psi(r)\le 2/9, \ 0 <r<1$ is equivalent to
$$\ln x+\frac1x-1\le\frac{2}{9}(x-1)$$
or, equivalently,
$$\ln x \le \frac{2}{9}(x-1)+1-\frac1x = \frac{(x-1)(2x+9)}{9x}$$
Hence $\psi(r) \leq \frac 29, \ 0<r<1$ is equivalent to $f(x)\ge0$, $x>1$, where
$$f(x)=\frac{(x-1)(2x+9)}{9x}-\ln x, \qquad f'(x)=\frac{(x-3)(2x-3)}{9x^2}$$
Extending $f$ continuously to $x=1$, we see that its minimum on $[1,\infty)$ is attained at $x=1$
or $x=3$. Since $f(1)=0$ and $f(3)=10/9-\ln3>0$, we obtain $\psi(r)\le 2/9$.
\newline
Consequently, each summand in \eqref{eq:h-prefix} is at most
$$\int_{C_{j-1}^*}^{C_j^*}\ln\frac Tt\,dt+\frac29p_j$$
Recall that $C_m^*\le T$ is the last optimal completion time not exceeding $T$.
Summing over the corresponding jobs,
$$\int_0^T h(t)\,dt
\le\int_0^{C_m^*}\ln\frac Tt\,dt+\frac29C_m^*
\le T+\frac29T=\frac{11}{9}T$$
\end{proof}
\begin{theorem}
The greedy algorithm satisfies
$$G\ge\frac9{20}\operatorname{OPT}$$
\end{theorem}

\begin{proof}
We use the same general summation-by-parts device as Bansal
et al.~\cite{BansalEtAl2017} in their analysis of precedence relations:
prefix inequalities are combined with nonnegative differences of a
monotone sequence. Here, this transfers the preceding prefix bound to
an integral weighted by the greedy priority $q$.
Let $0=T_0<T_1<\cdots<T_n=P$ be the greedy completion times,
and let $q(t)=q_k$ on $[T_{k-1},T_k)$.
Since $q_1\ge\cdots\ge q_n\ge0$, setting $q_{n+1}=0$ gives
$$q(t)=\sum_{k=1}^n(q_k-q_{k+1}) \mathbb{I}_{\{t<T_k\}}$$
By the preceding lemma and \eqref{eq:q-integral},
\begin{align*}
\int_0^P q(t)h(t)\,dt
&=\sum_{k=1}^n(q_k-q_{k+1})
  \int_0^{T_k}h(t)\,dt\\
&\le\frac{11}{9}\sum_{k=1}^n(q_k-q_{k+1})T_k
=\frac{11}{9}G.
\end{align*}
Using \eqref{eq:gap-bound}, this implies
$\operatorname{OPT}-G\le \frac {11}9G$ and $G\ge \frac 9 {20}\operatorname{OPT}$.
\end{proof}
\section{\texorpdfstring{An $O(n\log n)$ Implementation}{An O(n log n) Implementation}}\label{sec:implementation}
Direct computation of the maximum in \eqref{eq:greedy} leads to an $O(n^2)$ implementation. We show that the greedy order can be obtained in $O(n\log n)$ time.
\newline
We first restrict attention to the positive-weight jobs
$\Theta_+=\{i:w_i>0\}$.\footnote{\label{fn:zero-weight}If $w_i=0$, then $\phi_i(S)=0$ for every
$S\ge0$, whereas every remaining positive-weight job has positive priority.
Thus zero-weight jobs can be omitted from the geometric data structure and
appended at the end in arbitrary order. Their contribution to the objective
is zero. For definiteness, the implementation appends them in increasing
index order.}
For these jobs, maximizing $\phi_i(S)$ is equivalent to minimizing
$1/\phi_i(S)$:
\begin{equation}
    \frac{1}{\phi_i(S)} = \frac{p_i}{w_i}S + \frac{p_i^2}{w_i}
    \label{eq:line}
\end{equation}
For each positive-weight job $i$, define the point
\begin{equation}
    q_i=(a_i,b_i) = \left( \frac{p_i}{w_i}, \frac{p_i^2}{w_i} \right)
    \label{eq:point}
\end{equation}
Since $\langle q_i,(S,1)\rangle=p_i(S+p_i)/w_i=1/\phi_i(S)$ and
$\phi_i(S)>0$, the greedy choice among the remaining positive-weight jobs
$R\subseteq\Theta_+$ satisfies
\begin{equation}
    \mathop{\mathrm{arg\,max}}_{i\in R} \phi_i(S) = \mathop{\mathrm{arg\,min}}_{i\in R} \langle q_i,(S,1)\rangle
    \label{eq:extreme}
\end{equation}
The minimum of a linear functional over a finite set of points is attained at an extreme point of their convex hull. Hence, every greedy step can be interpreted as an extreme-point query on the current convex hull in the direction $-(S,1)$.\footnote{If several jobs attain the maximum priority, any of them is an admissible
greedy choice, and the approximation proof does not require strict
inequalities. The theoretical algorithm may return any maximizing job;
it does not require a smallest-index tie-breaking rule.}
Indeed, maximizing $\langle q_i,(-S,-1)\rangle$ is equivalent to minimizing
$\langle q_i,(S,1)\rangle$, so the query returns a job admissible under the
original greedy rule.
\newline
After a job is selected at each step of the greedy algorithm, it is removed
from the remaining jobs. Coincident points are grouped: each distinct point
$q$ stores a list $L(q)$ of jobs with those coordinates, and the point is
deleted from the hull structure only when its list becomes empty.
Grouping by lexicographically sorting the coordinates takes $O(n\log n)$
time and $O(n)$ space, while removing a job from its list takes $O(1)$ time.
Thus we need a deletion-only planar convex-hull structure. Using such a structure gives $O(n\log n)$ preprocessing, $O(\log n)$ extreme-point queries, and $O(\log n)$ amortized deletions.
Figure~\ref{fig:convex-hull-greedy-steps} illustrates the geometric interpretation of the greedy rule by a sequence of simple hull configurations.  
At each step the supporting line in direction $(-S,-1)$ touches the current
lower hull at a point corresponding to a greedy choice.
After the selected point is deleted, the hull changes and new points can become
vertices of the updated hull.  
This example illustrates why a static convex hull is not sufficient and why deletions must be supported dynamically.

\begin{figure}[!htbp]
    \centering

    \begin{subfigure}[t]{0.45\textwidth}
        \centering
        \includegraphics[width=\linewidth]{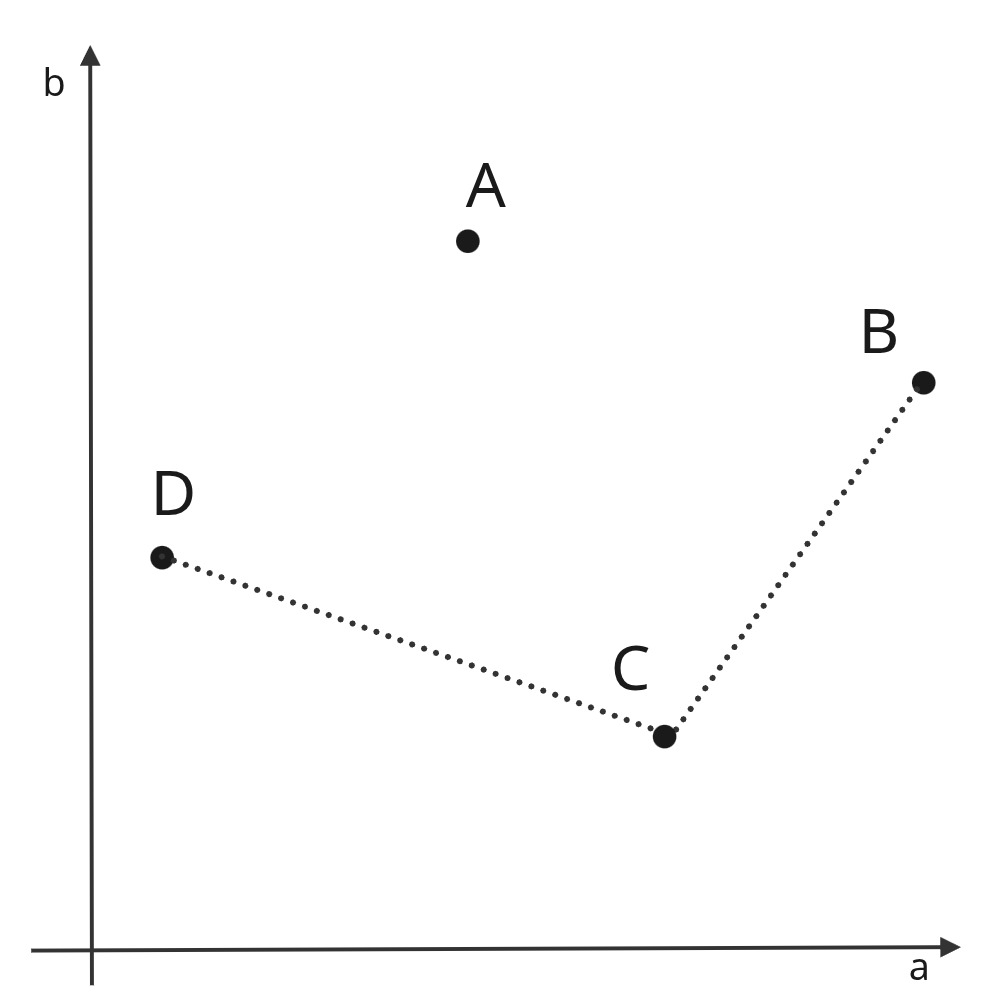}
        \caption{Initial lower hull.}
        \label{fig:convex-step1}
    \end{subfigure}
    \hfill
    \begin{subfigure}[t]{0.45\textwidth}
        \centering
        \includegraphics[width=\linewidth]{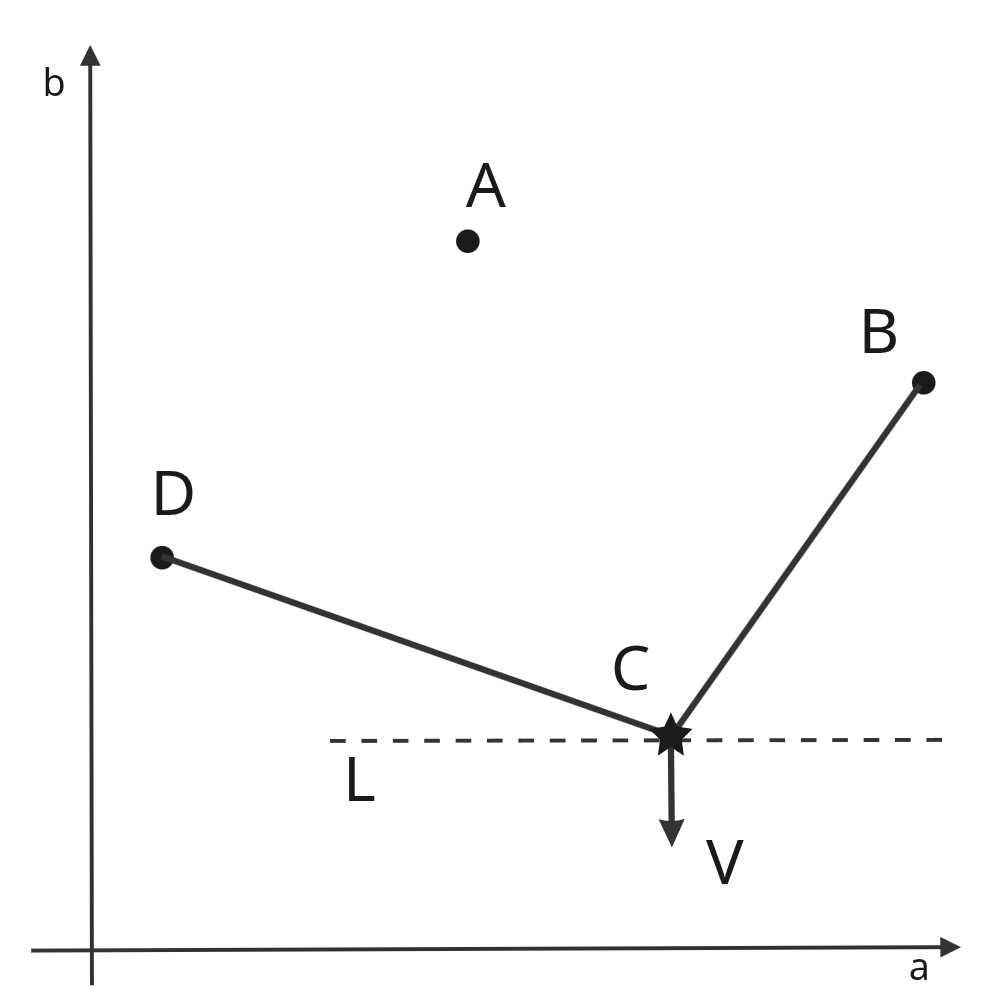}
        \caption{The current query selects $C$.}
        \label{fig:convex-step2}
    \end{subfigure}

    \vspace{0.35em}

    \begin{subfigure}[t]{0.45\textwidth}
        \centering
        \includegraphics[width=\linewidth]{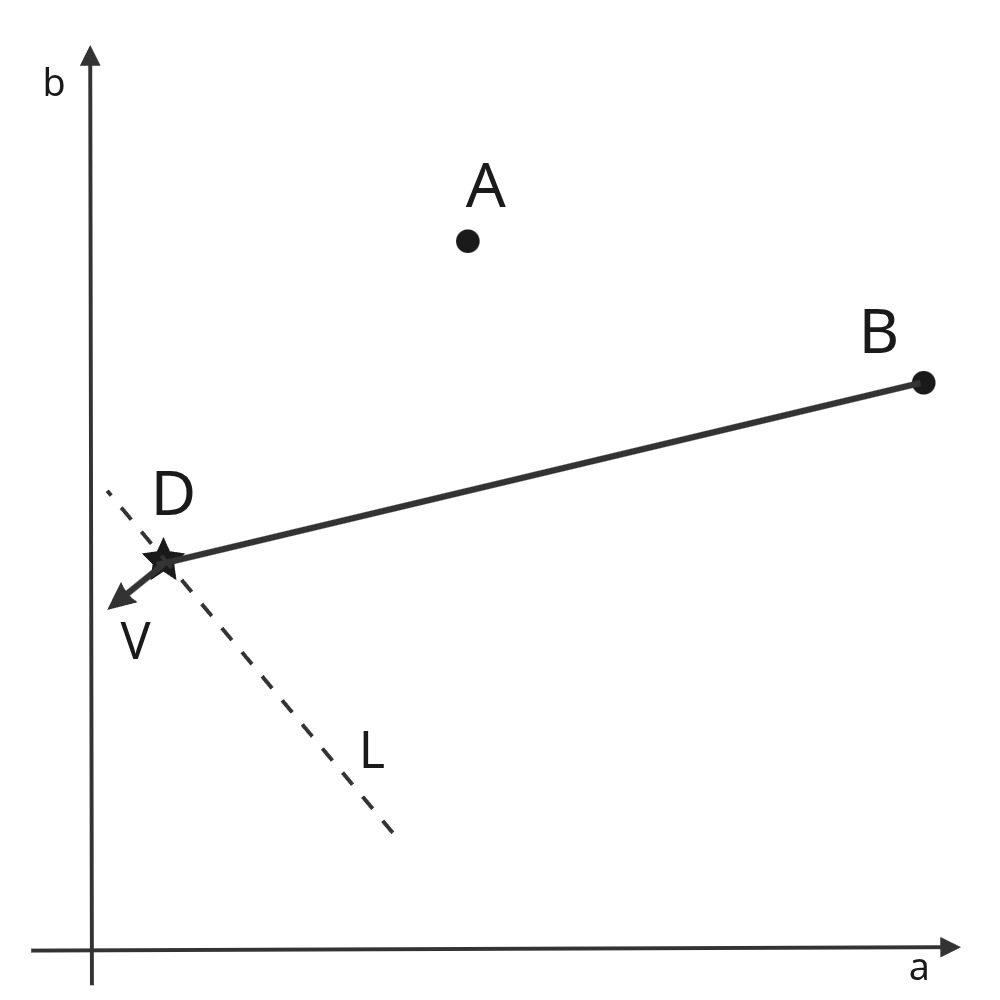}
        \caption{After deleting $C$, the hull changes.}
        \label{fig:convex-step3}
    \end{subfigure}
    \hfill
    \begin{subfigure}[t]{0.45\textwidth}
        \centering
        \includegraphics[width=\linewidth]{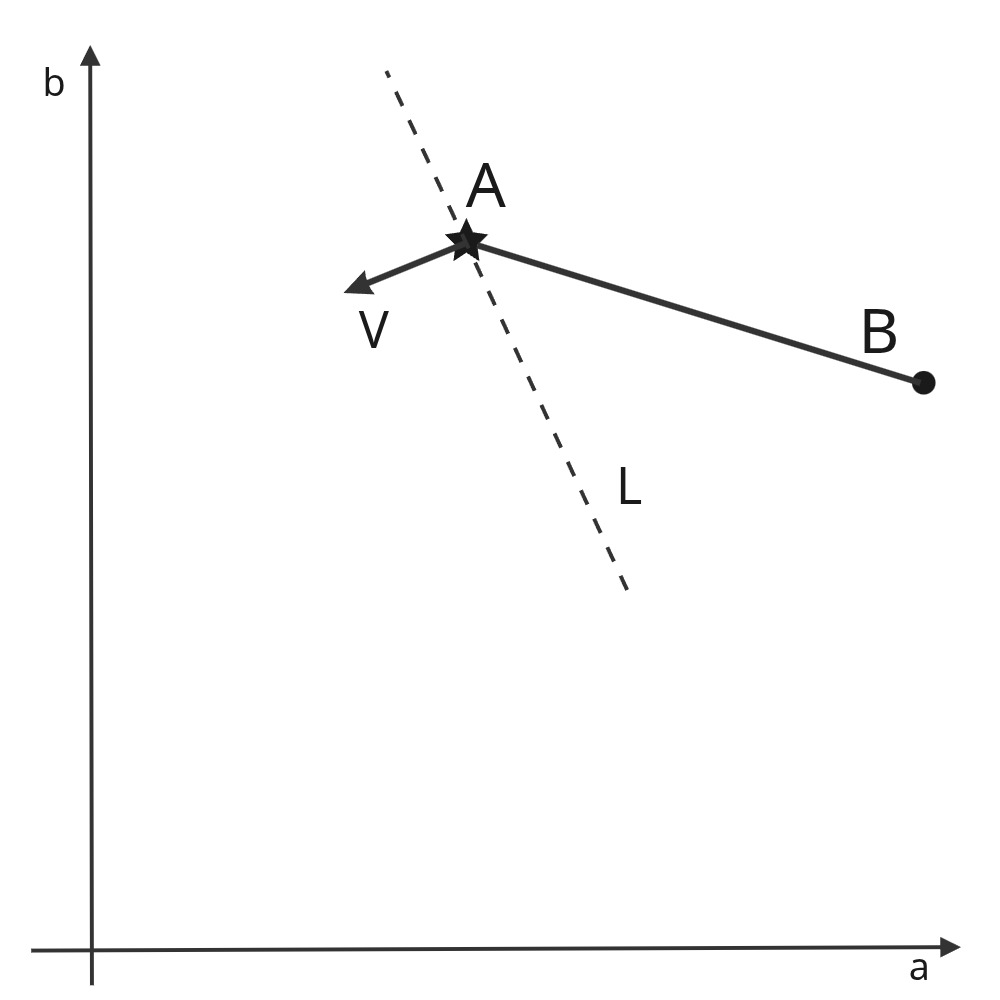}
        \caption{Later queries use the updated hull.}
        \label{fig:convex-step4}
    \end{subfigure}

    \caption{Geometric illustration of the greedy rule.
    Each point represents a job after the transformation
    $q_i=(p_i/w_i,p_i^2/w_i)$. The vector $V=(-S,-1)$ defines the
    supporting-line query for the current accumulated processing time $S$.
    The selected extreme point is marked by a star.}
    \label{fig:convex-hull-greedy-steps}
\end{figure}

\begin{theorem}[Fast implementation]
The greedy algorithm that, at the current value of $S$, selects
$$a\in \mathop{\mathrm{arg\,max}}_{i\in R} \frac{w_i} {p_i(S+p_i)}$$
can be implemented in $O(n\log n)$ time and $O(n)$ space in the real-RAM model.
\end{theorem}

\begin{proof}
Zero-weight jobs are handled separately as described in Footnote~\ref{fn:zero-weight}.
\newline
For every
remaining job $i$, we construct the point $q_i$ according to
\eqref{eq:point}. By \eqref{eq:extreme}, at total processing time
$S$, the greedy choice is exactly an extreme-point query over the set of
remaining points in direction $(-S,-1)$. Once a job is selected, its
entry is removed from its list, and the corresponding point is deleted only
when that list becomes empty. Thus, throughout the algorithm the point
set undergoes deletions only; no insertions are required.
\newline
We therefore use the deletion-only semi-dynamic convex-hull data structure
of Hershberger and Suri~\cite{HershbergerSuri1992}. As Hershberger and Suri proved, this data structure
can be constructed in $O(n\log n)$ preprocessing time, supports
extreme-point queries in $O(\log n)$ worst-case time, and supports point
deletions, including the corresponding update of the convex hull, in
$O(\log n)$ amortized time.
The fully dynamic structure of Brodal and Jacob~\cite{BrodalJacob2002}
is also suitable and additionally supports insertions, which are unnecessary
here but may be useful in future online variants with arriving jobs.
\newline
The greedy algorithm performs one extreme-point query and one list removal
for each positive-weight job, and at most one geometric deletion per such job.
Hence, over at most $n$ iterations, the total
time spent on queries and updates is $O(n\log n)$. Together with the
initial preprocessing, this yields an overall running time of
$$O(n\log n)$$
For $n$ initial points, the data structure of
Hershberger and Suri uses $O(n)$ space. Therefore, the complete
implementation also requires $O(n)$ space.
\end{proof}

\begin{algorithm}[H]
\caption{Convex-hull implementation of the greedy algorithm}
\label{alg:greedy}
\begin{algorithmic}[1]
\Require Jobs $J=\{1,\ldots,n\}$ with $p_i>0$ and $w_i\ge0$
\Ensure A permutation $\pi$
\State $\Theta_+\gets\{i\in J:w_i>0\}$
\For{each $i\in\Theta_+$}
    \State $q_i\gets\left(\dfrac{p_i}{w_i},\dfrac{p_i^2}{w_i}\right)$
\EndFor
\State Group jobs into lists $L(q)$ by identical coordinates
\Comment{$O(n\log n)$}
\State Build the hull structure on the distinct points $q$
\Comment{$O(n\log n)$}
\State $S\gets0$
\State $\pi\gets()$
\While{$\Theta_+\neq\emptyset$}
    \State $q\gets\Call{ExtremePoint}{(-S,-1)}$
    \Comment{$O(\log n)$ worst-case}
    \State Remove a job $\psi^*$ from $L(q)$ \Comment{$O(1)$}
    \State Append $\psi^*$ to $\pi$
    \If{$L(q)$ is empty}
        \State \Call{Delete}{$q$} \Comment{$O(\log n)$ amortized}
    \EndIf
    \State $\Theta_+\gets\Theta_+\setminus\{\psi^*\}$
    \State $S\gets S+p_{\psi^*}$
\EndWhile
\State Append all zero-weight jobs to $\pi$ in increasing index order
\State \Return $\pi$
\end{algorithmic}
\end{algorithm}

\paragraph*{Pseudocode.}
Algorithm~\ref{alg:greedy} summarizes the resulting implementation.
At each iteration, the next job is obtained by an extreme-point query in
direction $V=(-S,-1)$. One job is removed from the selected point's list,
the point is deleted if the list becomes empty, and the accumulated
processing time $S$ is updated.
\newline
The resulting schedule coincides with the schedule produced by the original greedy rule, up to tie-breaking.
Code for experimenting with the greedy algorithm is available at
\url{https://github.com/DR-LLL/ARP-Greedy}.

\section{Conclusions}\label{sec:conclusions}

We proved a $9/20$-approximation guarantee for an existing adaptive greedy
rule for ARP. A reduction to deletion-only convex-hull queries yields an
$O(n\log n)$-time, $O(n)$-space implementation in the real-RAM model.
This operation count is independent of the magnitudes of the input values,
but does not bound the bit complexity of exact arithmetic.
\newline
The greedy value $G$ gives the bounds
$$G\le\operatorname{OPT}\le\frac{20}{9}G.$$
It can therefore initialize an exact search with a feasible solution and
a certified upper bound. Any subsequent search that retains the best
feasible solution found preserves the $9/20$ guarantee even if interrupted.
\newline
The factor $9/20$ is a proved guarantee, not an established tight
worst-case approximation ratio. Determining the exact worst-case ratio
and evaluating the algorithm within exact search methods remain directions
for further research.

\section*{Statements and Declarations}

\noindent\textbf{Competing interests.}
The authors declare that they have no competing interests.

\noindent\textbf{Funding.}
The authors received no specific funding for this work.

\noindent\textbf{Data availability.}
No datasets were generated or analysed during the current study.

\noindent\textbf{Code availability.}
Code for experimenting with the greedy algorithm is available at
\url{https://github.com/DR-LLL/ARP-Greedy}.

\noindent\textbf{Author contributions.}
Both authors contributed to the development of the results and to the writing
of the manuscript.

\bibliography{sn-bibliography}

@book{GamowStern1958,
  author    = {Gamow, George and Stern, Marvin},
  title     = {Puzzle-Math},
  publisher = {Viking Press},
  address   = {New York},
  year      = {1958}
}

@article{Franklin1960,
  author  = {Franklin, J. N.},
  title   = {The Range of a Fleet of Aircraft},
  journal = {Journal of the Society for Industrial and Applied Mathematics},
  year    = {1960},
  volume  = {8},
  number  = {3},
  pages   = {541--548},
  doi     = {10.1137/0108039}
}

@article{Goldman1961,
  author  = {Goldman, A. J.},
  title   = {On the Range of a Fleet of Aircraft},
  journal = {Journal of Research of the National Bureau of Standards, Section B: Mathematics and Mathematical Physics},
  year    = {1961},
  volume  = {65B},
  number  = {4},
  pages   = {237--238}
}

@inproceedings{Woeginger2010,
  author    = {Anderson, Jim and Andersson, Bj{\"o}rn and Azar, Yossi and Bansal, Nikhil and Bini, Enrico and Chrobak, Marek and Correa, Jos{\'e} and Cucu-Grosjean, Liliana and Davis, Rob and Easwaran, Arvind and Edmonds, Jeff and Funk, Shelby and Gopalakrishnan, Sathish and Hoogeveen, Han and Mathieu, Claire and Megow, Nicole and Naor, Seffi and Pruhs, Kirk and Queyranne, Maurice and Ros{\'e}n, Adi and Schabanel, Nicolas and Sgall, Ji{\v r}{\'i} and Sitters, Ren{\'e} and Stiller, Sebastian and Uetz, Marc and Vredeveld, Tjark and Woeginger, Gerhard J.},
  title     = {10071 Open Problems -- Scheduling},
  booktitle = {Scheduling},
  series    = {Dagstuhl Seminar Proceedings},
  volume    = {10071},
  publisher = {Schloss Dagstuhl -- Leibniz-Zentrum fuer Informatik},
  address   = {Dagstuhl, Germany},
  year      = {2010},
  pages     = {1--24},
  doi       = {10.4230/DagSemProc.10071.3}
}

@article{GamzuSegev2019,
  author  = {Gamzu, Iftah and Segev, Danny},
  title   = {A Polynomial-Time Approximation Scheme for the Airplane Refueling Problem},
  journal = {Journal of Scheduling},
  year    = {2019},
  volume  = {22},
  number  = {1},
  pages   = {119--135},
  doi     = {10.1007/s10951-018-0569-x}
}

@article{Vasquez2015,
  author  = {V{\'a}squez, {\'O}scar C.},
  title   = {For the Airplane Refueling Problem Local Precedence Implies Global Precedence},
  journal = {Optimization Letters},
  year    = {2015},
  volume  = {9},
  number  = {4},
  pages   = {663--675},
  doi     = {10.1007/s11590-014-0758-2}
}

@inproceedings{LiEtAl2019,
  author    = {Li, Jianshu and Hu, Xiaoyin and Luo, Junjie and Cui, Jinchuan},
  title     = {A Fast Exact Algorithm for Airplane Refueling Problem},
  booktitle = {Combinatorial Optimization and Applications},
  series    = {Lecture Notes in Computer Science},
  volume    = {11949},
  pages     = {316--327},
  publisher = {Springer},
  address   = {Cham},
  year      = {2019},
  doi       = {10.1007/978-3-030-36412-0_25}
}

@article{Smith1956,
  author  = {Smith, Wayne E.},
  title   = {Various Optimizers for Single-Stage Production},
  journal = {Naval Research Logistics Quarterly},
  year    = {1956},
  volume  = {3},
  number  = {1--2},
  pages   = {59--66},
  doi     = {10.1002/nav.3800030106}
}

@inproceedings{SchulzUdwani2019,
  author    = {Schulz, Andreas S. and Udwani, Rajan},
  title     = {Robust Appointment Scheduling with Heterogeneous Costs},
  booktitle = {Approximation, Randomization, and Combinatorial Optimization.
               Algorithms and Techniques ({APPROX/RANDOM} 2019)},
  series    = {Leibniz International Proceedings in Informatics (LIPIcs)},
  volume    = {145},
  pages     = {1--17},
  note      = {Article 25},
  publisher = {Schloss Dagstuhl -- Leibniz-Zentrum fuer Informatik},
  address   = {Dagstuhl, Germany},
  year      = {2019},
  doi       = {10.4230/LIPIcs.APPROX-RANDOM.2019.25}
}

@article{GmeinerSchulz2026,
  author  = {Gmeiner, Simon and Schulz, Andreas S.},
  title   = {Block Stacking, Airplane Refueling, and Robust Appointment Scheduling},
  journal = {The American Mathematical Monthly},
  year    = {2026},
  volume  = {133},
  number  = {6},
  pages   = {540--559},
  doi     = {10.1080/00029890.2026.2644834}
}

@article{HershbergerSuri1992,
  author  = {Hershberger, John and Suri, Subhash},
  title   = {Applications of a Semi-Dynamic Convex Hull Algorithm},
  journal = {BIT Numerical Mathematics},
  year    = {1992},
  volume  = {32},
  number  = {2},
  pages   = {249--267},
  doi     = {10.1007/BF01994880}
}

@article{LinHe2026,
  author  = {Lin, Hao and He, Cheng},
  title   = {Precedence Criteria and Gradient-Based Scheduling Algorithm for the Airplane Refueling Problem},
  journal = {Chinese Quarterly Journal of Mathematics},
  year    = {2026},
  volume  = {41},
  number  = {1},
  pages   = {38--49},
  doi     = {10.13371/j.cnki.chin.q.j.m.2026.01.003}
}

@article{BansalEtAl2017,
  author  = {Bansal, Nikhil and D{\"u}rr, Christoph and Thang, Nguyen Kim and V{\'a}squez, {\'O}scar C.},
  title   = {The Local--Global Conjecture for Scheduling with Non-linear Cost},
  journal = {Journal of Scheduling},
  year    = {2017},
  volume  = {20},
  number  = {3},
  pages   = {239--254},
  doi     = {10.1007/s10951-015-0466-5}
}

@inproceedings{BrodalJacob2002,
  author    = {Brodal, Gerth St{\o}lting and Jacob, Riko},
  title     = {Dynamic Planar Convex Hull},
  booktitle = {Proceedings of the 43rd Annual IEEE Symposium on Foundations of Computer Science},
  pages     = {617--626},
  publisher = {IEEE},
  address   = {Los Alamitos, CA},
  year      = {2002},
  doi       = {10.1109/SFCS.2002.1181985}
}

@article{MegowVerschae2018,
  author  = {Megow, Nicole and Verschae, Jos{\'e}},
  title   = {Dual Techniques for Scheduling on a Machine with Varying Speed},
  journal = {SIAM Journal on Discrete Mathematics},
  year    = {2018},
  volume  = {32},
  number  = {3},
  pages   = {1541--1571},
  doi     = {10.1137/16M105589X}
}

\end{document}